\documentclass[twoside,11pt]{article}

\usepackage{jmlr2e,amsmath,bbm,bibentry,natbib}

\usepackage{amsthm}

\newtheorem{thm}{Theorem}[section]
\newtheorem{lem}[thm]{Lemma}
\newcommand\numberthis{\addtocounter{equation}{1}\tag{\theequation}}

\newcommand{\be}{\begin{equation}}
\newcommand{\ee}{\end{equation}}

\newcommand{\mh}{\mathcal{H}}
\def\HS{\mathrm{HS}}
\def\h{\mathcal{H}}
\def\mf{\mathcal{F}}
\def\P{\mathbb{P}}
\def\E{\mathbb{E}}

\usepackage{color}

\ShortHeadings{Feature Selection via KCCA}{Tianqi Liu, Kuang-Yao Lee, and Hongyu Zhao}
\firstpageno{1}

\begin{document}

\title{Ultrahigh Dimensional Feature Selection via Kernel Canonical Correlation Analysis}

\author{\name Tianqi Liu \email tianqi.liu@yale.edu \\
       \addr Department of Biostatistics\\
       Yale University\\
       New Haven, CT 06511, USA
       \AND
       \name Kuang-Yao Lee \email kuang-yao.lee@yale.edu \\
       \addr Department of Biostatistics\\
       Yale University\\
       New Haven, CT 06511, USA
       \AND
       \name Hongyu Zhao \email hongyu.zhao@yale.edu \\
       \addr Department of Biostatistics\\
       Yale University\\
       New Haven, CT 06511, USA}

\editor{Kevin Murphy, Bernhard Sch\"{o}lkopf}

\maketitle

\begin{abstract}%   <- trailing '%' for backward compatibility of .sty file
High-dimensional variable selection is important in many scientific fields, such as genomics. In this paper, we develop a Sure Independence feature Screening procedure based on Kernel Canonical Correlation Analysis (KCCA-SIS, for short).  No model assumption is needed between
response and predictors to apply KCCA-SIS and it can be used in ultrahigh dimensional
data analysis. Compared to the original SIS {\citep{fan2008sure}}, KCCA-SIS can handle nonlinear dependencies among variables. Compared to Distance Correlation-SIS {\citep{li2012feature}}, KCCA-SIS is scale free, distribution free and has better approximation results based on the universal characteristic of Gaussian Kernel {\citep{micchelli2006universal}}. KCCA-SIS encompasses SIS and DC-SIS in the sense that SIS and DC-SIS correspond to specific kernel choices under KCCA-SIS. Compared to sup-HSIC-SIS {\citep{balasubramanian2013ultrahigh}}, KCCA-SIS is scale-free removing the marginal variation of features and response variables. Similar to DC-SIS and sup-HSIC-SIS, KCCA-SIS can also be used directly to screen grouped predictors and handle multivariate response variables. We show that KCCA-SIS has the sure screening property, and has better performance through simulation studies and its application to a brain gene expression dataset. 
\end{abstract}


\section{Introduction}

Ultrahigh dimensional data sets have become common in  many disciplines. For example, the reducing cost in microarrays and sequencing allows researchers to collect information on gene expression and sequence data at the whole genome level. A typical study may generate expression information from tens of thousands of genes (denoted as $p$) across dozens to hundreds of subjects (denoted as $n$). Feature screening is important in genetics/genomics studies to identify disease genes, construct gene networks, and develop biomarkers. Various regularization methods have been proposed and their statistical properties studied for these high dimensional problems, such as: Lasso \citep{tibshirani1996regression}, Dantzig selector \citep{candes2007dantzig}, SCAD \citep{fan2001variable}, and MCP \citep{zhang2010nearly}. All of these methods allow the number of selected predictors to be larger than sample size. 

However, the above mentioned methods may not perform well for ultrahigh dimensional data due to the simultaneous challenges in computational efficiency, statistical consistency and algorithmic robustness (\citet{zhao2006model}, \citet{fan2009ultrahigh}, \citet{fan2010selective}). In order to tackle these difficulties, \citep{fan2008sure} proposed the Sure Independence Screening (SIS) and showed that the Pearson correlation ranking procedure possesses a sure screening property for linear regressions with Gaussian predictors and responses. Since the publication of SIS, several extensions were made to consider generalized linear models \citep{fan2009ultrahigh} and nonparametric independence screening in sparse ultrahigh dimensional additive models \citep{fan2011nonparametric}. \citet{ji2012ups} further proposed a two-stage method called UPS: screening by univariate thresholding and cleaning by penalized least squares for selecting variables. \citet{li2012feature} proposed DC-SIS, a sure independence screening model-free method based on distance correlation as a measure of relationship between response and covariate. \citet{song2012feature} proposed a method based on Hilbert--Schmidt Independence Criterion (HSIC, for short). To generalized the idea of DC-SIS, \citet{balasubramanian2013ultrahigh} proposed a general framework, called {sup}-HSIC-SIS, for model-free and multi-output screening. Motivated from the equivalence between distance covariance and HSIC \citep{sejdinovic2013equivalence}, they used Reproducing Kernel Hilbert Space (RKHS) based independence measures \citep{gretton2005measuring}. 

In this paper, we propose a new method called Kernel Canonical Correlation Analysis (KCCA)-SIS, which removes the marginal effect of variables compared to sup-HSIC-SIS and DC-SIS. HSIC  calculates the maximum covariance between the transformations of two random variables restricted in certain function classes, while KCCA calculates the maximum correlation between the transformed ones by removing the marginal variations of random variables. KCCA (\citet{akaho2006kernel}, \citet{melzer2001nonlinear}, \citet{bach2003kernel}) was first proposed as a nonlinear extension of canonical correlation aiming to extract the shared information between two random variables, i.e., to provide nonlinear mappings $f\in \mathcal{H}_X$ and $g\in\mathcal{H}_Y$ so that ${\mathrm{cor}[f(X),g(Y)]}$ is maximized. It was shown in \citet{fukumizu2007statistical} that the maximum of the objective function in KCCA is identical to the operator norm of the correlation operator between $\mathcal{H}_X$ and $\mathcal{H}_Y$. This fact motivates us to use the operator norm of the correlation operator as a measure for the relationship between random variables. We show that KCCA-SIS enjoys the sure screening property under mild conditions. In both simulations and a real data application for extracting interneuron related genes in the human brain, we show that the proposed method performs better than the existing approaches.

The rest of this paper is organized as follows. In Section 2, we develop the KCCA-SIS for feature screening and establish its sure screening property. In Section 3, we compare the proposed method with other approaches on simulated and real data. We conclude this paper with a brief discussion in Section 4. All technical proofs are given in the Appendix.

\def\HSIC{\mathrm{HSIC}}

\section{Independence screening using Kernel CCA}
\def\to{\rightarrow}
\subsection{Some Preliminaries}
Let $(\mathcal{X},\mathcal{B}_{\mathcal{X}})$ and $(\mathcal{Y},\mathcal{B}_{\mathcal{Y}})$ denote Borel measurable spaces. For example, they can be $\mathbb{R}^d$ or any topological Borel measurable spaces. Given positive definite kernels $k_x$ and $k_y$, let $(\mathcal{H}_X,k_x)$ and $(\mathcal{H}_Y,k_y)$ be RKHSs \citep{aronszajn1950theory} of functions on $\mathcal{X}$ and $\mathcal{Y}$, respectively. We denote the marginal distributions of $X$ and $Y$ as $\P_X$ and $\P_Y$, and their joint distribution as $\P_{XY}$. We denote the expectation operator associated with $\P_X$, $\P_Y$, and $\P_{XY}$ as $\E_X$, $\E_Y$, and $\E_{XY}$, respectively. For a random variable $X:\Omega\rightarrow \mathcal{X}$, the \emph{mean element} $m_X\in \mathcal{H}_X$ is induced by the relation, for all $f\in \mathcal{H}_X$
\begin{eqnarray*}
\langle f,m_X\rangle_{\mathcal{H}_X} = \E_X[\langle k_x(\cdot,X),f \rangle] = \E_X f(X),
\end{eqnarray*}
where $\langle \cdot,\cdot \rangle_{\mathcal{H}_X} $ denotes the inner product under ${\mathcal{H}_X}$. By the Riesz representation theorem \citep{reed1980methods}, there exists an operator {$\Sigma_{YX}:\mathcal{H}_X \to \mathcal{H}_Y$} so that
$$
\langle g,\Sigma_{YX}f\rangle_{\mathcal{H}_Y} = \E_{XY}[\langle f,k_x(\cdot,X)-m_X\rangle_{\mathcal{H}_X}\langle k_y(\cdot,Y)-m_Y,g\rangle_{\mathcal{H_Y}}] = \mathrm{Cov}(f(X),g(Y))
$$
holds for all $f\in \mathcal{H}_X$ and $g\in \mathcal{H}_Y$. We call this operator \emph{cross-covariance operator} \citep{fukumizu2009kernel}. If $Y$ is equal to $X$, the positive self-adjoint operator $\Sigma_{XX}$ is called the \emph{covariance operator}. \citet[Theorem 1]{baker1973joint} showed that $\Sigma_{YX}$ can be expressed as
\be
\Sigma_{YX} = \Sigma_{YY}^{1/2}\mathcal{R}_{YX}\Sigma_{XX}^{1/2},
\ee
where $\mathcal{R}_{YX}:\mathcal{H}_X\rightarrow \mathcal{H}_Y$ is a unique bounded operator such that $||\mathcal{R}_{YX}||\leq 1$. We call $\mathcal{R}_{YX}$ the \emph{correlation operator} from $\mathcal{H}_{X}$ to $\mathcal{H}_{Y}$, capturing all the nonlinear information between $X$ and $Y$. On the other hand, assuming $k:(\mathcal{X}\times\mathcal{Y})^2\rightarrow\mathbb{R}$ to be separable, i.e., $k((x,y),(x',y')) = k_x(x,x')k_y(y,y')$, where $k_x:\mathcal{X}^2\rightarrow \mathbb{R}$ and $k_y:\mathcal{Y}^2\rightarrow \mathbb{R}$ are reproducing kernels of $\mathcal{H}_X$ and $\mathcal{H}_Y$ respectively (in which case $\mathcal{H}$ is homomorphism to the tensor product of $\mathcal{H}_X$ and $\mathcal{H}_Y$. i.e., $\mathcal{H}\cong \mathcal{H}_X\otimes\mathcal{H}_Y$), the Hilbert--Schmidt independence criterion (HSIC) is defined as ${\mathrm{HSIC}}(\P_{XY},\mathcal{H}_X,\mathcal{H}_Y):=||\Sigma_{XY}||_{\HS}^2$, where $||\cdot||_{{\HS}}$ denotes the Hilbert--Schmidt norm of the operator. HSIC was first introduced by \citet{gretton2005measuring} and the authors showed that it can be represented as:
\begin{align*}
{\HSIC}(\P_{XY},\mathcal{H}_X,\mathcal{H}_Y)=&\E_{XX'YY'}[k_x(X,X')k_y(Y,Y')] + \E_{XX'}[k_x(X,X')]\E_{YY'}[k_y(Y,Y')]\\
& - 2\E_{XY}[\E_{X'}[k_x(X,X')]\E_{Y'}[k_y(Y,Y')]],
\end{align*}
where $(X',Y')$ are an independent copy of $(X,Y)$ and $\E_{XX'YY'}$ denotes the expectation over the independent pairs. Under the condition that $k_x$ and $k_y$ are characteristic \citep{fukumizu2007kernel}, ${\HSIC}(\P_{XY},\mathcal{H}_X,\mathcal{H}_Y)$ is zero iff $X$ and $Y$ are independent. From this, we know that $||\Sigma_{YX}|| = 0$ iff $X$ and $Y$ are independent, where $||\cdot||$ denotes the operator norm. Furthermore, it is easy to show that $||\mathcal{R}_{YX}|| = 0$ iff $X$ and $Y$ are independent \citep{fukumizu2007kernel}. 

With a slight abuse of notation, we write $\mathcal{R}_{YX} = \Sigma_{YY}^{-1/2}\Sigma_{YX}\Sigma_{XX}^{-1/2}$, where $\Sigma_{YY}$ and $\Sigma_{XX}$ may not be invertible. We define the regularized version of $\mathcal{R}_{YX}$ as 
$$
\mathcal{R}_{YX}(\epsilon_n) {\triangleq} (\Sigma_{YY} + \epsilon_nI)^{-1/2}\Sigma_{YX}(\Sigma_{XX} + \epsilon_nI)^{-1/2},
$$
where $\epsilon_n>0$ is the ridge parameter
\citet[][Lemma 7]{fukumizu2007statistical} showed that if ${\mathcal{R}}_{YX}$ is compact,
\begin{align*}
{||\mathcal{R}_{YX}(\epsilon_n)-{\mathcal{R}}_{YX}||\rightarrow 0, \ \mbox{as $\epsilon_n\rightarrow 0$}.}
\end{align*}

Next we derive a sample level estimator of $\mathcal{R}_{YX}(\epsilon_n)$. {Suppose $\{(X^{(i)},Y^{(i)})\}_{i=1}^n$ is a set of $n$ independent copies from $(X,Y)$. Then} the \emph{empirical cross-covariance operator} $\hat{\Sigma}_{YX}^{(n)}$ is defined as the cross-covariance operator under the empirical distribution $\frac{1}{n}\sum_{i=1}^n\delta_{X^{(i)}}\delta_{Y^{(i)}}$, where $\delta_{X^{(i)}}$ and $\delta_{Y^{(i)}}$ are Dirac measures with point mass on $X^{(i)}$ and $Y^{(i)}$. That is, for any $f\in \mathcal{H}_X$ and $g\in \mathcal{H}_Y$, $\hat{\Sigma}_{YX}^{(n)}$ satisfies
\begin{align*}
\langle g,\hat{\Sigma}_{YX}^{(n)}f\rangle_{\mathcal{H}_Y} = {\mathrm{Cov}}_n[f(X),g(Y)],
\end{align*}
where ${\mathrm{Cov}}_n(X,Y)$ is the empirical covariance between two random variables with respect to the empirical measure. We can similarly define $\hat{\Sigma}_{YY}^{(n)}$ and $\hat{\Sigma}_{XX}^{(n)}$. We then have the regularized estimator of $\mathcal{R}_{YX}$: 
\begin{align*}
&\hat{\mathcal{R}}_{YX}^{(n)}(\epsilon_n) \triangleq (\hat{\Sigma}_{YY}^{(n)} + \epsilon_n I)^{-1/2}\hat{\Sigma}_{YX}^{(n)}(\hat{\Sigma}_{XX}^{(n)} + \epsilon_n I)^{-1/2}.
\end{align*}
Empirically, we use $||\hat{\mathcal{R}}_{YX}^{(n)}(\epsilon_n)||$ as the measure of dependency between predictor $X$ and response $Y$. $\hat{\mathcal{R}}_{YX}^{(n)}(\epsilon_n)$ was first introduced in \citet{fukumizu2007statistical} and is called the normalized cross-covariance operator (NOCCO)

\subsection{An Independence Ranking and Screening Procedure}
In this section we propose an independence screening procedure based on KCCA. We assume a response $Y\in\mathbb{R}^d$ and predictors $X\in\mathbb{R}^{p}$, with $p$ growing with $n$ and $d$ fixed. It is often assumed that only a small number of predictors are relevant to $Y$. 

Denote by $\P(Y|X)$ the conditional distribution of $Y$ given $X$. Following \citep{li2012feature}, we define the set of relevant variables called \emph{active set} $\mathcal{M}$ and irrelevant variables called \emph{inactive set} $\mathcal{I}$ as:
\begin{align*}
\mathcal{M} &= \{r:\P(Y|X)\text{ depends on }X_r\}\text{, and}\\
\mathcal{I} &=\{r:\P(Y|X)\text{ does not depend on }X_r\}{.}
\end{align*}
We write $X_{\mathcal{M}} = \{X_r:r\in\mathcal{M}\}$ and $X_{\mathcal{I}} = \{X_r:r\in\mathcal{I}\}$, and call $X_{\mathcal{M}}$ as an \emph{active predictor vector} and its complement $X_{\mathcal{I}}$ as an \emph{inactive predictor vector}. By the definition we know that $Y$ and $X_{\mathcal{I}}$ are independent conditional on $X_{\mathcal{M}}$. In this case, feature selection involves estimating the set $\mathcal{M}$ from the given $n$ samples.

A direct way is to rank the predictors according to their degree of dependence with the response. We consider the norm of correlation operator as a measure of such dependence. To be specific, we write
$$
\rho_r(\epsilon_n) = ||(\Sigma_{YY} + \epsilon_nI)^{-1/2}\Sigma_{YX_r}(\Sigma_{X_rX_r}+\epsilon_nI)^{-1/2}||,
$$
to be the measure of dependence between $X_r$ and $Y$, because $\rho_r(\epsilon_n) = 0$ for any $\epsilon_n > 0$ iff $X_r$ and $Y$ are independent. Similar to distance correlation, our measure here is model-free and allows for multivariate response and group predictors. Similar to sup-HSIC-SIS, our method can be used in the case of more general topological space for response $Y$. 

\subsection{The learning algorithm}
\subsubsection{Choice of kernel}
As mentioned before, we choose Gaussian kernel for its universal property. The form of Gaussian kernel is defined as:
$$
k(x,y) = \exp(-\gamma\|x-y\|_2^2),
$$
where $\|\cdot\|_2$ stands for Euclidean norm.

In sample version, we have the corresponding estimator $\hat{\rho}_r (\epsilon_n)= ||\hat{\mathcal{R}}^{(n)}_{YX_r}(\epsilon_n)||$. In order to select the relevant variables, we first compute $\hat{\rho}_r(\epsilon_n)$ for $r = 1,...,p$ and define
$$
\hat{\mathcal{M}} = \{r:\hat{\rho}_r(\epsilon_n)\geq C_3\epsilon_n^{-3/2}n^{-\kappa}, \text{ for }1\leq r\leq p\}
$$
as the estimated set of active predictors, where $0\leq \kappa < 1/2$, $C_3$ is predefined constant in condition (C2) and $\epsilon_n^{-3/2}$ is due to some technical issues explained later.

\subsubsection{Sample level estimator}
Following \citet{lee2016variable}, we will derive the empirical representation of $||\hat{\mathcal{R}}^{(n)}_{YX_r}(\epsilon_n)||$, where $Y\in\mathbb{R}^d$ and $X_r\in\mathbb{R}$. Suppose we observe $n$ i.i.d samples $(X_r^{(1)},Y^{(1)}),...,(X_r^{(n)},Y^{(n)})$, let $K_{X_r},K_Y$ be two positive semidefinite kernel matrices with $(K_{X_r})_{ij} = k({X_r}^{(i)},{X_r}^{(j)})$ and $(K_Y)_{ij} = k(Y^{(i)},Y^{(j)})$. Let $Q = I_n - \frac{1}{n}\bold{1}\bold{1}^T$, $G_{X_r} = QK_{X_r}Q$, and $G_Y = QK_YQ$. Let the singular value decompositions of $G_{X_r}$ and $G_Y$ be $U_{X_r}D_{X_r}U_{X_r}^T$ and $U_YD_YU_Y^T$, respectively. Here $U_{X_r}, D_{X_r}, U_Y, D_Y\in\mathbb{R}^{n\times n}$. {We use $A^ \dagger$ to denote the Moore--Penrose inverse of a matrix $A$, and $A^ {\dagger \alpha}$ to denote $(A^ \dagger ) ^ \alpha$ } . We choose the orthonormal basis {$$(\phi_1,...,\phi_{r_x}) = (k_{x_r}(\cdot,{X_r}^{(1)}),...,k_{x_r}(\cdot,{X_r}^{(n)}))QU_{X_r}D_{X_r}^{{\dagger 1/2}}$$ and 
$$
(\psi_1,...,\psi_{r_y}) = (k_y(\cdot,Y^{(1)}),...,k_y(\cdot,Y^{(n)}))QU_YD_Y^{{\dagger 1/2}}.
$$}  
 Then we can represent $f = (\phi_1,...,\phi_{r_x})[f]$ for $[f]\in\mathbb{R}^n$ and
 \be\label{eq:repf}
 (f({X_r}^{(1)}),...,f({X_r}^{(n)}))^T = K_{X_r}QU_{X_r}D_{X_r}^{{\dagger 1/2}}[f].
 \ee
 The notation $[\cdot]$ is the coordinate with respect to the new basis system; \citet{lee2013general} and \citet{lee2016variable} also adopted a similar coordinate system. We denote $\mathcal{H}^{(n)}_{X_r}\subseteq \mathcal{H}_{X_r}$ to be the RKHS generated by $(k_{x_r}(\cdot,{X_r}^{(1)}),...,k_{x_r}(\cdot,{X_r}^{(n)}))$ and similarly for $\mathcal{H}^{(n)}_Y\subseteq \mathcal{H}_Y$. Then for any two functions $f_1,f_2\in\mathcal{H}^{(n)}_{X_r}$
\begin{align*}
\langle f_1,f_2\rangle_{\mathcal{H}_{X_r}} &= [f_1]^TD_{X_r}^{{\dagger 1/2}}U_{X_r}^TQK_{X_r}QU_{X_r}D_{X_r}^{{\dagger 1/2}}[f_2] = [f_1]^T[f_2].
\end{align*}
For $f\in \mathcal{H}^{(n)}_{X_r}$ and $g\in \mathcal{H}^{(n)}_Y$, 
\begin{align*}
[g]^T[\hat{\Sigma}^{(n)}_{Y{X_r}}][f] = \langle g,\hat{\Sigma}^{(n)}_{Y{X_r}}f\rangle_{\mathcal{H}_Y} & = (g(Y_1),...,g(Y_n))^TQ(f({X_r}_1),...,f({X_r}_n))\\
&=[g]^TD_Y^{{\dagger 1/2}}U_Y^TQK_YQQK_{X_r}QU_{X_r}D_{X_r}^{{\dagger 1/2}}[f]\\
&=[g]^TD_Y^{1/2}U_Y^TU_{X_r}D_{X_r}^{1/2}[f],
\end{align*}
where the second equality follows from equation (\ref{eq:repf}). So we have $[\hat{\Sigma}^{(n)}_{Y{X_r}}] = D_Y^{1/2}U_Y^TU_{X_r}D_{X_r}^{1/2}$, $[\hat{\Sigma}^{(n)}_{YY}] = D_Y$, $[\hat{\Sigma}^{(n)}_{{X_r}{X_r}}] = D_{X_r}$. Then we can easily show that
\be\label{eq:R}
[\hat{\mathcal{R}}_{Y{X_r}}^{(n)}(\epsilon_n)] = (D_Y+\epsilon_nI)^{-1/2}D_Y^{1/2}U_Y^TU_{X_r}D_{X_r}^{1/2}(D_{X_r} + \epsilon_nI)^{-1/2}.
\ee
Since we just conduct the orthogonal transformation of the original matrix, the operator norm of sample correlation operator is just the largest singular value of $[\hat{\mathcal{R}}_{Y{X_r}}^{(n)}]$. 

\subsubsection{Tuning parameter selection}
For Gaussian kernel, we need to choose the bandwidth parameter $\gamma$.  {For $i=1,\ldots,p$, we compute $\gamma_i$ via
\begin{equation}\label{width}
\frac{1}{\sqrt{\gamma_ i}} = \frac{2\sqrt{2}}{n(n-1)}\sum_{i<j}\|X^{(i)} - X^{(j)}\|_2. 
\end{equation}
Similarly we can compute $\gamma_Y$ for $Y$}.

For the choice of $\epsilon_n$, we use a generalized cross-validation (GCV) criterion similar to \citet{li2014additive}. To be specific, let $L_Y = (\bold{1},K_Y)^T$, $L_r = (\bold{1},K_{X_r})^T$, where $K_Y$ and $K_{X_r}$ are the corresponding kernel matrices. Then we define
\be\label{eq:GCV}
{\mathrm{GCV}}(\epsilon_n) = \sum_{r=1}^p\frac{||L_Y-L_YL_{r}^T(L_{r}L_{r}^T + \epsilon_n I_{n+1})^{-1}L_r||_F^2}{\{1-{\mathrm{tr}}(L_r^T(L_rL_r^T + \epsilon_nI_{n+1})^{-1}L_r)/n\}^2},
\ee
where $\|\cdot\|_F$ is the Frobenius norm of a matrix. We choose $\epsilon_n$ by minimizing $\mathrm{GCV}(\epsilon_n)$.

\subsubsection{Feature screening procedure}
The algorithm is as follows:
{
\begin{enumerate}

\item[(a)] Calculate the bandwidth parameters $\gamma_1,\ldots,\gamma_p$, and $\gamma_Y$ using (\ref{width});

\item[(b)] Calculate the ridge parameter $\epsilon_ n$ determined by (\ref{eq:GCV}) by grid search in the set $\{10^{-5},10^{-4},...,10^3\}$;

\item[(c)] Compute the gram matrices $G_Y$, $G_{X_1}$, ... , $G_{X_p}$ based on the Gaussian kernel function, and find their singular value decompositions;

\item[(d)] Compute the norm of $[\hat{\mathcal{R}}_{YX_i}^{(n)}(\epsilon_n)]$ based on (\ref{eq:R}); 

\item[(e)] Rank $\|[\hat{\mathcal{R}}_{YX_i}^{(n)}(\epsilon_n)]\|$ for $i=1,\ldots,p$. Suppose $\|[\hat{\mathcal{R}}_{YX_{r_1}}^{(n)}(\epsilon_n)]\|\geq \cdots \geq \|[\hat{\mathcal{R}}_{YX_{r_p}}^{(n)}(\epsilon_n)]\|$; we then estimate $\mathcal{M}$ by $\hat{\mathcal{M}} = \{r_1,\ldots,r_m\}$.\end{enumerate}
}
In practice, the choice of $m$ may depend on the researchers' prior knowledge and also the sample size. In our simulation analysis, we use different numbers of $m$ based on the true number of active predictors. In our real data analysis, we choose the upper $1\%$ as active predictors. Empirically, we recommend using $1.5\epsilon_n^{-3/2}n^{1/4}$, where $\epsilon_n$ is the best tuning parameter chosen by (\ref{eq:GCV}).

\subsection{Theoretical Guarantees}
In this section, we study the theoretical properties of the proposed independence screening method. Our analysis does not require any moment conditions on the variables $X$ and $Y$ such as spherical symmetric distribution in \citet{fan2008sure}, or sub-gaussian in \citet{li2012feature}. {Instead, we require the following two conditions:}
\begin{enumerate}
\item[(C1)] The uniform boundedness of kernel functions:
\be
\sup_{1\leq r\leq p}k_{x_r}(x,x)\leq {B} < \infty,\quad k_y(y,y)\leq {B} < \infty
\ee
\item[(C2)] The minimum signal strength condition:
\be
\min_{r\in\mathcal{M}}\rho_r(\epsilon_n)\geq 2C_3\epsilon_n^{-3/2}n^{-\kappa}, 
\ee
for some constants $C_3>0$ specified in Theorem \ref{thm:3} and $0\leq\kappa <1/2$.
\end{enumerate}
Note that condition {(C1) holds for many commonly used kernels, such as the radial basis function.} Condition (C2) requires that KCCA measure corresponding to the active predictors cannot be too weak, which is an analog of  condition 3 of \citet{fan2008sure}. First, we have a concentration bound for cross-covariance operator as in Theorem \ref{thm:1}:
\begin{thm}\label{thm:1}
Suppose (C1) holds, then we have for $r = 1,...,p$,
\begin{align*}
\P\{ ||\hat{\Sigma}_{YX_r} -\Sigma_{YX_r} ||_{\HS} - \E||\hat{\Sigma}_{YX_r} -\Sigma_{YX_r} ||_{\HS} \geq t\}\leq \exp(-\frac{C_2nt^2}{B^2}){.}
\end{align*}
\end{thm}
Based on the concentration bound in Theorem \ref{thm:1}, we can establish the following concentration bound for the correlation operator:
\begin{thm}\label{thm:2}
Suppose (C1) holds, $\epsilon_n = o(1)$ and $n^{-1}\epsilon_n^{-3} = o(1)${. Then} there exist constants $C_1,C_2>0$, such that 
\begin{align*}
\P\{ ||\hat{{\mathcal{R}}}_{YX_r}^{(n)}(\epsilon_n) - (\Sigma_{YY} + \epsilon_nI)^{-1/2}\Sigma_{YX_r}&(\Sigma_{X_rX_r} + \epsilon_nI)^{-1/2}|| \\
&- C_1K^{3/2}\epsilon_n^{-3/2}n^{-1/2}\geq t\}\leq 3\exp(-\frac{C_2\epsilon_n^3 nt^2}{B^2}),
\end{align*}
for $r = 1,...,p$.
\end{thm}

Based on the concentration bounds and conditions (C1) and (C2), we can achieve the following sure independence screening property.
\begin{thm}\label{thm:3}
Suppose (C1) holds, there exist constants $C_3,C_4>0$, $ 0\leq \kappa < 1/2$ such that
\begin{align*}
\P\{\max_{1\leq r\leq p}|\hat{\rho}_r (\epsilon_n)- \rho_r(\epsilon_n)|\geq C_3B^{3/2}\epsilon_n^{-3/2}n^{-\kappa}\}\leq 3p\exp(-{C_4Bn^{1-2\kappa}}){.}
\end{align*}
Furthermore if condition (C2) holds, then we have the following sure screening property:
\begin{align*}
\P\{\mathcal{M}\subseteq \hat{\mathcal{M}}\}\geq 1 - 3s\exp(-{C_4Bn^{1-2\kappa}}),
\end{align*}
where $s$ is the cardinality of $\mathcal{M}$.
\end{thm}
Based on the above result, we can handle the NP dimensionality $\log p = o(n^{1-2\kappa})$. 

The sure screening property without controlling for false selection rates is not satisfactory. Ideally if there is a gap between active variables and inactive variables regarding their $\rho(\epsilon_n)$, i.e. $\max_{j\notin \mathcal{M}}\rho_r(\epsilon_n) = o(B^{3/2}\epsilon_n^{3/2}n^{-\kappa})$, the false-positive rate will vanish.
 
{Next, we show that the size of $\hat{\mathcal{M}}$ can be controlled when there is no severe dependency between the predictors.
Suppose $\h_X$ is the direct sum $\oplus_{r=1}^p\h_{X_r}$; in other words, $\h_X$ is induced by the additive kernel  $k_x(s,t) = \sum_{i=1}^pk_{x_r}(s_i,t_i)$, for any $s=(s_1,\ldots,s_p)$ and $t=(t_1,\ldots,t_p)$. It can be shown that the covariance operator $\Sigma_{XX}: \h_X \to \h_X$ has a matrix form satisfying that, for any $f=(f_1,\ldots,f_p)\in \h_X$,
\begin{eqnarray*}
\Sigma_{XX} f = \sum_{r=1}^p \sum_{s=1}^p\Sigma_{X_rX_s}f_s.
\end{eqnarray*}
Then the following result provides an upper bound for $|\hat{\mathcal{M}}|$.
\begin{thm}\label{thm:4}
For $\epsilon_n \leq 1$, we have 
$$
\P\{|\hat{\mathcal{M}}|\leq O(n^{2\kappa}\lambda_{\mathrm{max}}(\Sigma_{XX}))\} \geq 1 - 3p\exp(-{C_4Bn^{1-2\kappa}}),
$$
where $\lambda_{\mathrm{max}}(\cdot)$ represents the largest singular value of the corresponding operator, and $C_4>0$ is the constant in Theorem \ref{thm:3}.
\end{thm}
}

\section{Numerical Results}
In this section, we report results on different simulated and real biological data to illustrate the advantage of the propose method (KCCA-SIS). For the experiments on synthetic data, we consider the data settings from \citet{li2012feature} and \citet{balasubramanian2013ultrahigh} in order to make a {head} to head comparison to their approaches. For evaluation on real world data, we consider a high dimensional human brain gene expression data set, select genes related to marker genes for interneuron cells, and measure the performance of the selection using gene set enrichment analysis.

In simulations 1 and 2, we generate random vector $X = (X_1,X_2,\cdots,X_p)$ from a multivariate Gaussian distribution with zero mean and covariance matrix $\Sigma = (\sigma_{ij})_{p\times p}$, where $\sigma_{ij} = 0.8^{|i - j|}$. The error term $\varepsilon$ is generated from $N(0,1)$. We fix the sample size $n$ to be 200 and number of features $p$ to be 2000. We repeat each experiment 500 times, and evaluate the performance through the following two criteria (the same as those used in \citet{li2012feature}).
\begin{enumerate}
\item[1]
$\mathcal{S}$: the minimum model size to include all active predictors. We report the $ 25\%, 50\%,$ and $75\%$ quantiles of $\mathcal{S}$ using replications.
\item[2]
$\mathcal{P}$: the proportion that all active predictors are selected for a given model size $d$ in the 500 replications.
\end{enumerate}

The metric $\mathcal{S}$ is used as a measure of model complexity needed for sure screening with regard to the underlying screening procedure. The lower the value of $\mathcal{S}$, the better the screening procedure. The sure screening property ensures that $\mathcal{P}$ is close to one when the estimated model size $d$ is sufficiently large. We choose $d$ to be $d_1 = [n/\log n]$, $d_2 = 2d_1$ and $d_3 = 3d_1$ throughout our simulations, where $[c]$ denotes the integer part of $c$.
\subsection{Simulation 1}
This example is designed to compare the finite sample performance of the KCCA-SIS with SIS \citep{fan2008sure}, DC-SIS \citep{li2012feature} and HSIC-SIS \citep{balasubramanian2013ultrahigh}. We generate the response $Y$ according to four models (The first three models are used in \citet{li2012feature}):
\begin{enumerate}
\item[1.]

$Y = c_1\beta_1X_1X_2 + c_3\beta_2\mathbbm{1}(X_{12} < 0) + c_4\beta_3 X_{22} + \varepsilon$;

\item[2.]
$Y = c_1\beta_1X_1X_2 + c_3\beta_2\mathbbm{1}(X_{12} < 0)X_{22} + \varepsilon$;
\item[3.]
$Y = c_1\beta_1X_1 + c_2\beta_2X_2 + c_3\beta_3\mathbbm{1}(X_{12} < 0) + \exp(c_4|X_{22}|)\varepsilon$;
\item[4.]
$Y = X_1/X_2 + X_{12}^2/(1 + \cos (X_{22})) + \varepsilon$,
\end{enumerate}
where $\beta_j = (-1)^U(a + |Z|)$, $a = 4\log n/\sqrt{n}$, $U\sim$ Bernoulli(0.4) and $Z\sim N(0,1)$. We set $(c_1,c_2,c_3,c_4) = (2,0.5,3,2)$ in this example. For each independence screening procedure, we compute the associated marginal effect of $X_r$ on $Y$. In this case we treat $X = (X_1,....,X_p)$ as the predictor variables. We use the GCV criterion to select $\epsilon_n$. 

Tables \ref{Tab:1} and \ref{Tab:2} report the simulation results for $\mathcal{S}$ and $\mathcal{P}$. We can observe that screening fails in all four models by SIS. The proposed method outperforms DC-SIS in all cases and HSIC-SIS in most cases. We notice that our proposed KCCA-SIS is better than DC-SIS, comparable with sup-HSIC-SIS in model 3, where there is heteroscedasticy. The better performance is likely due to the removal of the marginal variations of responses and predictors. We have similar results as HSIC-SIS  for larger $\epsilon_n$. The advantage of the proposed approach is clearly demonstrated in model 4, where the marginal variations are different among predictors. In that case KCCA-SIS performs much better than the other methods. 

\medskip

\begin{table}[ht]
\centering
\resizebox{\columnwidth}{!}{
\begin{tabular}{c|ccc|ccc|ccc|ccc}
\hline
$\mathcal{S}$ & \multicolumn{3}{c}{SIS} & \multicolumn{3}{c}{DC-SIS} & \multicolumn{3}{c}{HSIC-SIS} & \multicolumn{3}{c}{KCCA-SIS} \\ \hline
Model & $25\%$& $50\%$ & $75\%$&$25\%$& $50\%$ & $75\%$&$25\%$& $50\%$ & $75\%$&$25\%$& $50\%$ & $75\%$\\\hline
1& 208.3& 818.0& 1534.0                           & 8.0 & 13.0&20.0                 		    & 7.0 &10.0&	16.0		               & 5.0 &7.0&11.0 \\\hline
2& 801.8& 1302.0& 1663.5                           & 11.0 & 16.0&41.5                 		    & 6.0 &8.0&	13.0		               & 5.0 &6.0&9.0 \\\hline
3& 581.0& 1135.0& 1598.0                           & 7.0 & 13.0&60.3                 		    & 5.0 &8.0&	17.0		               & 6.0 &8.0&27.0 \\\hline
4& 1534.0& 1807.0& 1924.3                           & 385.8 & 770.5&1174.0                 		    & 52.0 &358.0&	867.0		               & 33.0 &139.0&463.3 \\\hline
\end{tabular}
}
\vspace{-.1in}
\caption{ Minimum model size ($\mathcal{S}$) comparisons among different methods in simulation 1}
\label{Tab:1}
\end{table}

\begin{table}[ht]
\begin{center}
\begin{tabular}{c|ccc|ccc|ccc|ccc}
\hline
$\mathcal{P}$ & \multicolumn{3}{c}{SIS} & \multicolumn{3}{c}{DC-SIS} & \multicolumn{3}{c}{HSIC-SIS} & \multicolumn{3}{c}{KCCA-SIS} \\ \hline
Model & $d_1$& $d_2$ & $d_3$&$d_1$& $d_2$ & $d_3$&$d_1$& $d_2$ & $d_3$&$d_1$& $d_2$ & $d_3$\\\hline
1& 0.08& 0.14& 0.17                           & 0.90 & 0.96&0.97                 		    & 0.92 &0.95&0.97		               & 0.94 &0.96&0.97 \\\hline
2& 0.00& 0.01& 0.02                          & 0.73 &0.86&0.91                 		    & 0.92 &0.95&	0.96		               & 0.95 &0.97&0.98 \\\hline
3& 0.01& 0.03& 0.05                           & 0.70 & 0.77&0.80                		    & 0.84 &0.88&	0.90		               & 0.78 &0.85&0.87 \\\hline
4& 0.00& 0.00& 0.00                           & 0.00 & 0.01&0.04                 		    & 0.06 &0.12&0.20		               & 0.21 &0.30&0.37 \\\hline
\end{tabular}
\end{center}
\vspace{-.1in}
\caption{ The proportions ($\mathcal{P}$) comparisons among different methods in simulation 1}
\label{Tab:2}
\end{table}

\medskip

\subsection{Simulation 2}
In {this} experiment, we consider multivariate outputs, while $X$ is generated as before. We generate $Y|X\sim N(\bf{0},\Sigma)$ from a bivariate normal distribution, where $\sigma_{11} = \sigma_{22} = 1$ and $\sigma_{12} = \sigma_{21} = \sigma(X)$. We consider two correlation functions for $\sigma(X)$ given by
\begin{enumerate}
\item[1.]
$\sigma(X) = \sin (\beta_1^TX)$ where $\beta_1 = (0.8,0.6,0,...,0)$;
\item[2.]
$\sigma(X) = \{\exp (\beta_2^TX) - 1\} / \{\exp (\beta_2^TX) + 1\}$ where $\beta_2 = (2 - U_1,2 - U_2, 2 - U_3 , 2 - U_4,0,...,0)$ with $U_i$ drawn i.i.d. from Uniform[0,1].
\end{enumerate} 
In model 1, we choose $d_1 = 2$. In model 2, we choose $d_1 = 4$. And we choose $d_2 = 2d_1$ and $d_3 = 3d_1$ as before. The simulation settings are identical to those in \citet{li2012feature}. Since the response {is a vector}, SIS cannot be applied in this scenario. The simulation results are shown in Table \ref{Tab:3} and Table \ref{Tab:4}. 

\medskip

\begin{table}[ht]
\begin{center}
\begin{tabular}{c|ccc|ccc|ccc}
\hline
$\mathcal{S}$  & \multicolumn{3}{c}{DC-SIS} & \multicolumn{3}{c}{HSIC-SIS} & \multicolumn{3}{c}{KCCA-SIS} \\ \hline
Model &$25\%$& $50\%$ & $75\%$&$25\%$& $50\%$ & $75\%$&$25\%$& $50\%$ & $75\%$\\\hline
1                        & 3.0 &  7.0  &16.0                                    &2.0 &2.0& 3.0 &2.0 &2.0&2.0\\\hline
2                        & 4.0 & 5.0&7.0                 		    & 4.0 &4.0&	4.0		               & 4.0 &4.0&4.0 \\\hline
\end{tabular}
\end{center}
\vspace{-.1in}
\caption{Minimum model size ($\mathcal{S}$) comparisons among different methods in simulation 2}
\label{Tab:3}
\end{table}

\medskip

\begin{table}[ht]
\begin{center}
\begin{tabular}{c|ccc|ccc|ccc}
\hline
$\mathcal{P}$ & \multicolumn{3}{c}{DC-SIS} & \multicolumn{3}{c}{HSIC-SIS} & \multicolumn{3}{c}{KCCA-SIS} \\ \hline
Model &$d_1$& $d_2$ & $d_3$&$d_1$& $d_2$ & $d_3$&$d_1$& $d_2$ & $d_3$\\\hline
1& 0.170 &0.364 & 0.480     &   0.678 &0.868  & 0.926&            0.984& 0.996 &1.000 \\\hline
2& 0.488 &0.856 &0.930 & 0.768 &0.960 &0.984 & 0.978 &1.000 &1.000 \\\hline
\end{tabular}
\end{center}
\vspace{-.1in}
\caption{The proportions ($\mathcal{P}$) comparisons among different methods in simulation 2}
\label{Tab:4}
\end{table}

\subsection{Real data}
In this subsection, we analyze a brain spatial temporal gene expression data set from \citet{kang2011spatio}. We consider gene expression data from 10 neocortex areas (MFC, OFC, DFC, VFC, M1C, S1C, IPC, A1C, STC, ITC) at 13 developmental stages (early fetal to late adulthood). For each gene, there are $10\times 13 = 130$ observations corresponding to a spatial temporal characterization of this gene. There are a total of 17568 genes. \citet{zeisel2015cell} reported newly identified marker genes for interneuron cell types using single cell RNA sequencing on mouse brain. We use those marker genes, including SP8, POU3F4, TOX3, NPAS1, SOX6, NKX2-1, LHX6, PAX6, DLX5, ARX, DLX2, DLX1, ELAVL2, and SP9, as the response variable. Interneurons have been found to function in reflexes, neuronal oscillations, and neurogenesis in the adult mamalian brain \citep{kandel2000principles}. And \citet{zeng2012large} found that the interneuron marker genes are conserved between mouse and human, thus we apply the identified marker genes directly as response variables in human brain gene expression data set. Since this is a multivariate response, we can use DC-SIS, HSIC-SIS, and KCCA-SIS to select predictiors. We select the top 1 percent of genes (i.e., $|\hat{\mathcal{M}}|=176$) related to the above marker genes (including themselves), and then conduct gene enrichment analysis (http://geneontology.org/page/go-enrichment-analysis). We choose the union of five most significant biological processes. The results of fold-change and p-value related to biological processes for each method are shown in Table \ref{Tab:realData1}. We can see that KCCA-SIS captures more biologically meaningful genes as reflected in lower p-values. For neurogenesis, KCCA-SIS identifies 43 enriched genes, while DC-SIS identifies 36 and HSIC-SIS identifies 38 genes, respectively. KCCA-SIS is more powerful in selecting genes with similar biological functions. Besides, KCCA-SIS leads to 45 significant enrichment biological process terms, while DC-SIS leads to 16 terms and HSIC-SIS leads to 21 terms. This suggests that the results provided by KCCA-SIS are more biologically meaningful. 

\medskip

\begin{table}[ht]
\begin{center}
\begin{tabular}{c|c|c|c}
\hline
Biological Process &DC-SIS&HSIC-SIS&KCCA-SIS \\\hline
nervous system development & 6.56E-05 & 3.37E-06&9.68E-11\\\hline
central nervous system development &1.00E00 & 1.00E00&1.19E-10\\\hline
neurogenesis& 6.48E-04 &1.20E-04&4.80E-08\\\hline
single-multicellular organism process& 4.12E-06&2.51E-05&5.12E-08\\\hline
head development &1.00E00&1.00E00&1.02E-07\\\hline
multicellular organismal process & 5.14E-04&9.56E-04&5.07E-07\\\hline
anatomical structure development & 6.60E-04&1.51E-03&5.62E-06\\\hline
system development &1.09E-03&1.93E-04&9.90E-06\\\hline
regulation of biological process &3.75E-03&2.77E-04&7.01E-04\\\hline
\end{tabular}
\end{center}
\vspace{-.1in}
\caption{Gene Ontology enrichment analysis}
\label{Tab:realData1}
\end{table}

\section{Discussion}
In this article we have proposed an ultrahigh dimensional feature selection method via Kernel Canonical Correlation Analysis. The proposed approach is scale-free, model-free and works with multivariate random variables. We {established} the sure screening property of the proposed method and illustrated its capability in handling ultrahigh dimensional data on various simulated and real biological data sets. 

Future work includes a theoretical analysis of the choice of thresholding and combination of KCCA-SIS and other nonlinear regression methods for a better predictive model.

% Acknowledgements should go at the end, before appendices and references
\acks{We would like to acknowledge support for this project
from the Yale World Scholars Program sponsored by the China Scholarship Council, and National Institutes of Health grants R01 GM59507 and P01 CA154295 awarded to Hongyu Zhao.}


\newpage

\appendix
\section{}\subsection{Some useful lemmas}
\begin{lem}[\citet{fukumizu2007statistical}]\label{lem:1}
Suppose $A$ and $B$ are positive self-adjoint operators on Hilbert space such that $0\leq A\leq \lambda I$ and $0\leq B\leq \lambda I$ hold for a positive constant $\lambda$. Then,
\begin{align*}
||A^{3/2} - B^{3/2}||\leq 3\lambda^{1/2}||A-B||{.}
\end{align*}
\end{lem}

\begin{lem}[ \citet{fukumizu2007statistical}]\label{lem:gret}
The cross-covariance operator $\Sigma_{YX}$ is a Hilbert-Schmidt operator, and its Hilbert-Schmidt norm is given by 
\begin{align*}
&||\Sigma_{YX}||_{\HS}^2\\
&=\E_{YX}\E_{\tilde{Y}\tilde{X}}[\langle k_x(\cdot,X) - m_X, k_y(\cdot, \tilde{X}) - m_X\rangle_{\mathcal{H}_X}\langle k_y(\cdot,Y) - m_Y, k_y(\cdot, \tilde{Y}) - m_Y\rangle_{\mathcal{H}_Y}]\\
& = ||\E_{YX}[(k_x(\cdot,X) - m_X)(k_y(\cdot,Y) - m_Y)]||^2_{\mathcal{H}_X\otimes \mathcal{H}_Y}{,}
\end{align*}
where $(\tilde{X},\tilde{Y})$ and $(X,Y)$ are independently and identically distributed with distribution $P_{XY}$.
\end{lem}

Let's consider a fixed predictor $X_r$ first. Let's denote $F_r = k_{x_r}(\cdot,X_r) - \E_{X_r}[k_{x_r}(\cdot,X_r)], G = k_y(\cdot,Y) - \E_Y[k_y(\cdot,Y)]$.  For given i.i.d data $\{(X^{(i)},Y^{(i)})\}_{i=1}^n$, $F_{ri} = k_{x_r}(\cdot, X_r^{(i)}) - \E_{X_r}[k_{x_r}(\cdot, X_r)], G_i = k_y(\cdot,Y^{(i)}) - \E_Y[k_y(\cdot,Y)]$, and $\mathcal{F}_r = \mathcal{H}_{X_r}\otimes \mathcal{H}_Y$ with kernel $k((x,y),(x',y')) {=} k_{x_r}(x,x')k_y(y,y')$. Then, $F_r,F_{r1},...,F_{rn}$ are i.i.d random elements in $\mathcal{H}_{X_r}$, and a similar fact holds for $G,G_1,...,G_n$. Notice that mean elements can be written as $m_{X_r} = \E_{X_r}k_{x_r}(\cdot,X_r), m_Y = \E_Yk_y(\cdot,Y)$ \citep{fukumizu2007statistical}.

\begin{lem}\label{lem:4pre}
Under assumptions that $\sup k_{x_r}(x,x)\leq {B} < \infty,\quad k_y(y,y)\leq {B} < \infty$, we have for $r = 1,...,p$ and $i = 1,...,n$,
\begin{align*}
||F_{ri}||_{\mh_{X_r}}\leq 2\sqrt{B}, \quad&||G_i||_{\mh_Y}\leq 2\sqrt{B},\\
||F_{ri} - F'_{ri}||_{\mh_{X_r}}\leq 2\sqrt{B}, \quad&||G_i - G'_i||_{\mh_Y}\leq 2\sqrt{B}{,}\\
||m_{X_r}||_{\mh_{X_r}}\leq \sqrt{B}, \quad&||m_Y||_{\mh_Y}\leq \sqrt{B}.
\end{align*}

\end{lem}
\begin{proof}
$$
||F_{ri}||_{\mh_{X_r}} = ||k_{x_r}(\cdot,X_r) - m_{X_r}||_{\mh_{X_r}}\leq ||k_{x_r}(\cdot,X_r)||_{\mh_{X_r}} + ||m_{X_r}||_{\mh_{X_r}}\leq \sqrt{B} + \sqrt{B}= 2\sqrt{B},
$$
where the first inequality comes from triangle inequality and the second from the definition of $B$ and $||\cdot||_{\mathcal{H}_{X_r}}$. Using the similar techniques, we have
$$
||F_{ri} - F'_{ri}||_{\mh_X} = ||k_{x_r}(\cdot,X_r) - k_{x_r}(\cdot,X_r')||_{\mh_{X_r}}\leq ||k_{x_r}(\cdot,{X_r})||_{\mh_{X_r}} + ||k_{x_r}(\cdot,X_r')||_{\mh_{X_r}}\leq 2\sqrt{B}.
$$
By Cauchy-Schwartz inequality we have
\begin{align*}
||m_{X_r}||^2_{\mh_{X_r}} &= \langle \E_{X_r}k(\cdot,X_r),\E_{X_r'}k(\cdot,X_r')\rangle\leq \E_{X_r}\E_{X'_r}k(X_r,X_r)^{1/2}k(X_r',X_r')^{1/2}\\
&\leq (\E_{X_r}k(X_r,X_r))^{1/2}(\E_{X_r'}k(X_r',X_r'))^{1/2}\leq B{.}
\end{align*}
This completes the proof.
\end{proof}

\begin{lem}\label{lem:4}
Under assumptions that $\sup k_{x_r}(x,x)\leq {B} < \infty,\quad k_y(y,y)\leq {B} < \infty$, we have for $r = 1,...,p$,
$$
\E||\hat{\Sigma}^{(n)}_{Y{X_r}} - \Sigma_{Y{X_r}} ||_{{\HS}} \leq c_1Bn^{-1/2},
\E||\hat{\Sigma}^{(n)}_{{X_r}{X_r}} - \Sigma_{{X_r}{X_r}} ||_{{\HS}} \leq c_1Bn^{-1/2},
\E||\hat{\Sigma}^{(n)}_{YY} - \Sigma_{YY} ||_{{\HS}} \leq c_1Bn^{-1/2}
$$
for some positive constant $c_1$. And
\begin{align*}
||\Sigma_{Y{X_r}}||_{\HS}\leq 4{B},\quad||\Sigma_{{X_r}{X_r}}||_{\HS}\leq 4{B},\quad||\Sigma_{YY}||_{\HS}\leq 4{B},\\
||\hat{\Sigma}_{Y{X_r}}||_{\HS}\leq 8{B},\quad||\hat{\Sigma}_{{X_r}{X_r}}||_{\HS}\leq 8{B}, \quad||\hat{\Sigma}_{YY}||_{\HS}\leq 8{B}.
\end{align*}
\end{lem}
\begin{proof}
Following the same argument as in \citet{fukumizu2007statistical}, Lemma \ref{lem:gret} implies
\be\label{eq:empEst}
||\hat{\Sigma}_{Y{X_r}}^{(n)}||_{\HS}^2 = \left|\left| \frac{1}{n}\sum_{i=1}^n\left( F_{ri} - \frac{1}{n}\sum_{j=1}^nF_{rj} \right)\left(G_{i} - \frac{1}{n}\sum_{j=1}^nG_j\right) \right|\right|_{\mf_r}^2{.}
\ee
{Using the argument in the proof of the same lemma,}
\begin{align*}
\langle\Sigma_{Y{X_r}},\hat{\Sigma}_{Y{X_r}}^{(n)}\rangle_{\HS} = \left\langle \E[F_rG],\frac{1}{n}\sum_{i=1}^n\left(F_{ri} - \frac{1}{n}\sum_{j=1}^n F_{rj}\right)\left(G_{i} - \frac{1}{n}\sum_{j=1}^nG_j\right) \right\rangle_{\mf_r}{.}
\end{align*}
From these equations, we have
\begin{align}
||\hat{\Sigma}^{(n)}_{Y{X_r}} &- \Sigma_{Y{X_r}} ||_{\HS}^2 = ||\Sigma_{Y{X_r}}||_{\HS}^2 - 2\langle\Sigma_{Y{X_r}},\hat{\Sigma}_{Y{X_r}}^{(n)}\rangle_{\HS} + ||\hat{\Sigma}_{Y{X_r}}^{(n)}||_{\HS}^2\\
&=\left|\left| \frac{1}{n}\sum_{i=1}^n\left( F_{ri} - \frac{1}{n}\sum_{j=1}^nF_{rj} \right)\left(G_{i} - \frac{1}{n}\sum_{j=1}^nG_j\right) -\E[F_rG]\right|\right|_{\mf_r}^2\\\label{eq:diff}
&=\left|\left| \frac{1}{n}\sum_{i=1}^nF_{ri}G_{i}  - \E[F_rG] - \left(2 - \frac{1}{n}\right)\left(\frac{1}{n}\sum_{i=1}^nF_{ri}\right)\left(\frac{1}{n}\sum_{i=1}^nG_{i}\right)\right|\right|_{\mf_r}^2,
\end{align}
which is further bounded by
$$
\left|\left| \frac{1}{n}\sum_{i=1}^nF_{ri}G_{ri}  - \E[FG]\right|\right|_{\mf_r} + 2\left|\left|\left(\frac{1}{n}\sum_{i=1}^nF_{ri}\right)\left(\frac{1}{n}\sum_{i=1}^nG_{ri}\right)\right|\right|_{\mf_r}.
$$
Let $Z_{ri} = F_{ri}G_{i} - \E[F_rG]$. Since the variance of a sum of independent random variables is equal to the sum of their variances, we obtain
\be
\E\left|\left|\frac{1}{n}\sum_{i=1}^nZ_{ri}\right|\right|^2_{\mf_r} = \frac{1}{n}\E||Z_{r1}||_{\mf_r}^2.
\ee
\begin{align*}
\E||Z_{r1}||_{\mf_r}^2 &=\E||F_{r1}G_1 - \E[F_rG]||_{\mf_r}^2\\
&\leq 2\E||F_{r1}G_1||_{\mf_r}^2 + 2||\E[F_rG]||_{\mf_r}^2 \\
&\leq 2\E||F_{r1}||_{\mh_{X_r}}^2||G_1||_{\mh_Y}^2 + 2(\E||F_rG||_{\mf_r})^2\\
&\leq 4\E||F_{r1}||_{\mh_{X_r}}^2||G_1||_{\mh_Y}^2\\
&\leq64B^2
\end{align*}
The first inequality follows from the fact that $||a-b||^2\leq 2||a||^2 + 2||b||^2$. The second inequality follows from Jenson's inequality $||\E[F_rG]||_{\mathcal{F}_r}\leq \E||F_rG||_{\mathcal{F}_r}$. The third inequality follows from the fact that $(\E||F_rG||_{\mf_r})^2 \leq \E||F_rG||^2_{\mf_r}\leq \E||F_r||_{\mathcal{H}_{X_r}}^2||G||_{\mathcal{H}_Y}^2$. The last inequality follows from lemma \ref{lem:4pre} that $||F_{ri}||_{\mathcal{H}_{X_r}}\leq 2\sqrt{B}$ and $||G_{i}||_{\mathcal{H}_Y}\leq 2\sqrt{B}$.

From the inequalities
\begin{align*}
\E\left|\left|\left(\frac{1}{n}\sum_{i=1}^nF_{ri}\right)\left(\frac{1}{n}\sum_{i=1}^nG_i\right)\right|\right|_{\mf_r} &= \E\left[\left|\left| \frac{1}{n}\sum_{i=1}^nF_{ri} \right|\right|_{\mh_{X_r}}\left|\left| \frac{1}{n}\sum_{i=1}^nG_i \right|\right|_{\mh_Y}\right]\\
&\leq\left(\E\left|\left| \frac{1}{n}\sum_{i=1}^nF_{ri} \right|\right|^2_{\mh_{X_r}}\right)^{1/2}\left(\E\left|\left| \frac{1}{n}\sum_{i=1}^nG_i \right|\right|^2_{\mh_Y}\right)^{1/2}{,}
\end{align*}
{and}
$$
\E\left|\left| \frac{1}{n}\sum_{i=1}^nF_{ri} \right|\right|^2_{\mh_{X_r}} = \frac{1}{n}\E||F_{r1}||_{\mh_{X_r}}^2\leq \frac{4B}{n}{,}
$$
{w}e have 
$$
\E||\hat{\Sigma}^{(n)}_{Y{X_r}} - \Sigma_{Y{X_r}} ||_{\HS}\leq (64B^2)^{1/2}n^{-1/2} + 4Bn^{-1}\leq c_1Bn^{-1/2}  
$$
for some constant $c_1 > 0$. Following the same argument we can {show that}
$$
\E||\hat{\Sigma}^{(n)}_{{X_r}{X_r}} - \Sigma_{{X_r}{X_r}} ||_{\HS}\leq c_1Bn^{-1/2}  \ \text{, and} \
\E||\hat{\Sigma}^{(n)}_{YY} - \Sigma_{YY} ||_{\HS}\leq c_1Bn^{-1/2} {,}
$$

To prove part 2, we have by lemma \ref{lem:gret}
$$
||\Sigma_{Y{X_r}}||_{\HS}^2 = ||\E[F_rG]||_{\mf_r}^2 \leq (\E||F_r||_{\mh_{X_r}}||G||_{\mh_Y})^2\leq 16B^2,
$$
where the first inequality follows from Jenson's inequality with respect to $||\cdot||_{\mathcal{F}_r}$ and the fact that $||F_rG||_{\mathcal{F}_r} = ||F_r||_{\mh_{X_r}}||G||_{\mh_Y}$, and the last inequality follows from lemma \ref{lem:4pre}.
\begin{align*}
||\hat{\Sigma}_{Y{X_r}}^{(n)}||_{\HS}^2 &= \left|\left| \frac{1}{n}\sum_{i=1}^n\left( F_{ri} - \frac{1}{n}\sum_{j=1}^nF_{rj} \right)\left(G_i - \frac{1}{n}\sum_{j=1}^nG_j\right) \right|\right|_{\mf_r}^2\\
&\leq  2\left|\left| \frac{1}{n}\sum_{i=1}^nF_{ri}G_i\right|\right|^2 + 2 \left|\left| \frac{1}{n^2}\sum_{i=1}^nF_{ri}\sum_{i=1}^nG_i\right|\right|^2\\
& \leq 2\frac{1}{n^2}\sum_{i,j=1}^n\langle F_{ri},F_{rj}\rangle_{\mh_{X_r}}\langle G_i,G_j\rangle_{\mh_Y} + 2\frac{1}{n^4}\sum_{i,j,k,l=1}^n\langle F_{ri},F_{rj}\rangle_{\mh_{X_r}}\langle G_k,G_l\rangle_{\mh_Y}\\
& \leq 64B^2,
\end{align*}
where the first inequality comes from the fact that $||a + b||^2\leq 2||a||^2 + 2||b||^2$ and the last inequality follows from lemma \ref{lem:4pre}.
The proof arguments are similar for $\Sigma_{{X_r}{X_r}}$, $\hat{\Sigma}_{{X_r}{X_r}}$, $\Sigma_{YY}$, and $\hat{\Sigma}_{YY}$.
\end{proof}

\begin{lem}\label{lem:5}
Let $\epsilon_n$ be a positive number such that $\epsilon_n\rightarrow 0 (n\rightarrow \infty)$. Then, for the i.i.d. sample $\{(X^{(i)},Y^{(i)})\}_{i=1}^n$, we have for $r = 1,...,p$,
$$
\E||\hat{{\mathcal{R}}}_{YX_r}^{(n)} - (\Sigma_{YY} + \epsilon_nI)^{-1/2}\Sigma_{YX_r}(\Sigma_{X_rX_r} + \epsilon_nI)^{-1/2}||\leq c_2K^2\epsilon_n^{-3/2}n^{-1/2}
$$
for some positive constant $c_2>0$.
\end{lem}
\begin{proof}
Following the same argument as in \citet{fukumizu2007statistical}, {the difference $\hat{{\mathcal{R}}}_{YX}^{(n)} - (\Sigma_{YY} + \epsilon_nI)^{-1/2}\Sigma_{YX}(\Sigma_{XX} + \epsilon_nI)^{-1/2}$ can be} decomposed as
\begin{align*}
\hat{{\mathcal{R}}}_{YX}^{(n)}& - (\Sigma_{YY} + \epsilon_nI)^{-1/2}\Sigma_{YX}(\Sigma_{XX} + \epsilon_nI)^{-1/2}\\
 = \{(&\hat{\Sigma}_{YY}^{(n)}+\epsilon_nI)^{-1/2} - (\Sigma_{YY} + \epsilon_nI)^{-1/2}\} \hat{\Sigma}_{YX}^{(n)}(\hat{\Sigma}_{XX}^{(n)} + \epsilon_nI)^{-1/2}\\
&+(\Sigma_{YY}+\epsilon_nI)^{-1/2}\{\hat{\Sigma}_{YX}^{(n)}-\Sigma_{YX}\}(\hat{\Sigma}_{XX}^{(n)}+\epsilon_nI)^{-1/2}\\
&+(\Sigma_{YY}+\epsilon_nI)^{-1/2}\Sigma_{YX}\{ (\hat{\Sigma}_{XX}^{(n)}+\epsilon_nI)^{-1/2} - (\Sigma_{XX} + \epsilon_nI)^{-1/2} \}\numberthis \label{eq:1}\\
=& M_1 + M_2 + M_3
\end{align*}
Using the equality
\be\label{eq:DC}
D^{-1/2} - C^{-1/2} = C^{-1/2}(C^{3/2} - D^{3/2})D^{-3/2} + (D-C)D^{-3/2},
\ee
we can rewrite $M_1$ as
\begin{align*}
\{(\Sigma_{YY} + \epsilon_nI)^{-1/2}((\Sigma_{YY} + \epsilon_nI)^{3/2} - (\hat{\Sigma}_{YY}^{(n)}+\epsilon_nI)^{3/2} ) + (\hat{\Sigma}_{YY}^{(n)} -  \Sigma_{YY}) \}\\
\times(\hat{\Sigma}_{YY}^{(n)}+\epsilon_nI)^{-3/2} \hat{\Sigma}_{YX}^{(n)}(\hat{\Sigma}_{XX}^{(n)} + \epsilon_nI)^{-1/2},
\end{align*}
the norm of which is further upper bounded by
$$
\frac{1}{\epsilon_n}\left\{\frac{3}{\sqrt{\epsilon_n}}\max\{||\Sigma_{YY} + \epsilon_nI||^{1/2},||\hat{\Sigma}_{YY}^{(n)} + \epsilon_nI||^{1/2}\} + 1\right\}||\hat{\Sigma}_{YY}^{(n)}-\Sigma_{YY} ||.
$$
The upper bound comes from the fact that $|| (\Sigma_{YY} + \epsilon_nI)^{-1/2} ||\leq\epsilon_n^{-1/2}, (\hat{\Sigma}_{YY}^{(n)}+\epsilon_nI)^{-1/2} \hat{\Sigma}_{YX}^{(n)}(\hat{\Sigma}_{XX}^{(n)} + \epsilon_nI)^{-1/2}\leq 1$ \citep{fukumizu2007statistical}, and Lemma \ref{lem:1}, 

Provided that $\epsilon_n\rightarrow 0$, by Lemma \ref{lem:4} we have 
$$
\E||M_1|| \leq cB^{3/2}\epsilon_n^{-3/2}n^{-1/2}
$$
for some constant $c>0$.
Similarly we have $\E||M_3|| \leq cB^{3/2}\epsilon_n^{-3/2}n^{-1/2}$. From Lemma \ref{lem:4} and the fact that $|| (\Sigma_{YY} + \epsilon_nI)^{-1/2} ||\leq\epsilon_n^{-1/2}$, we know 
$$\E||M_2||\leq c'\epsilon_n^{-1}n^{-1/2}.$$ So we have for some constant $c_2>0$,
\begin{align*}
\E||\hat{{\mathcal{R}}}_{YX}^{(n)} - (\Sigma_{YY} + \epsilon_nI)^{-1/2}\Sigma_{YX}(\Sigma_{XX} + \epsilon_nI)^{-1/2}||\leq c_2K^{3/2}\epsilon_n^{-3/2}n^{-1/2}.
\end{align*}
{We then complete the proof of the lemma.}
\end{proof}

\begin{lem}[McDiarmid's Inequality (\citet{mcdiarmid1989method})]
Let $X_1,...,X_n$ be independent random variables taking values in a set $A$, and assume that $f;A^n\rightarrow \mathbb{R}$ satisfies
\begin{align*}
\sup_{x_1,...,x_n,x_i'\in A}\left|f(x_1,...,x_n) - f(x_1,...,x_{i-1},x'_i,x_{i+1},...,x_n)\right|\leq c_i
\end{align*}
for every $1\leq i\leq n$. Then, for every $t>0$,
\begin{align*}
\P\{f(X_1,...,X_n) - \E f(X_1,...,X_n)\geq t\}\leq e^{-2t^2/{\sum_{i=1}^nc_i^2}}{.}
\end{align*}
\end{lem}

\subsection{Proof of main theorems}
\begin{proof}[Proof of Theorem \ref{thm:1}]
It suffices to check the bounded difference property of $ ||\hat{\Sigma}_{YX_r} -\Sigma_{YX_r} ||_{\HS}$. 
Denote $f((X_r^{(1)},Y^{(1)}),...(X_r^{(n)},Y^{(n)})) = ||\hat{\Sigma}^{(n)}_{YX_r} - \Sigma_{YX_r} ||_{\HS}$. By equation (\ref{{eq:empEst}}) 
\begin{align*}
||\hat{\Sigma}_{YX}||_{\HS}&=  \left|\left| \frac{1}{n}\sum_{i=1}^n\left( F_i - \frac{1}{n}\sum_{j=1}^nF_j \right)\left(G_i - \frac{1}{n}\sum_{j=1}^nG_j\right) \right|\right|_{\mf} \\
&= \left|\left| \frac{1}{n}\sum_{i=1}^n F_i G_i - \frac{1}{n^2}\sum_{i,j=1}^nF_iG_j \right|\right|_{\mf},
\end{align*}
we have
\begin{align*}
&|f((X_r^{(1)},Y^{(1)}),...(X_r^{(n)},Y^{(n)})) - f((X_r^{(1)},Y^{(1)}),...,(X_r^{'(i)},Y^{'(i)}),...,(X^{(n)},Y^{(n)})) | \\
\leq &||\hat{\Sigma}^{(n)}_{YX_r}-\hat{\Sigma}^{(n)}_{YX_r'}||_{\HS}\\
= & ||\frac{1}{n}(F_{ri}G_i-F'_{ri}G'_i) - \frac{1}{n^2}\left\{\sum_{j\neq i}[(F_{ri}-F_{ri}')G_j + F_{rj}(G_i-G_i')] + (F_{ri}G_i - F_{ri}'G_i') \right\} ||_{\mf_r}\\
\leq &||\frac{1}{n}(F_{ri}G_i-F'_{ri}G'_i)||_{\mf} + \frac{1}{n^2}\left|\left|\left\{\sum_{j\neq i}[(F_{ri}-F_{ri}')G_j + F_{rj}(G_i-G_i')] + (F_{ri}G_i - F_{ri}'G_i') \right\} \right|\right|_{\mf_r}\\
\leq &\frac{8B}{n} + \frac{1}{n}(||(F_{ri}-F_{ri}')G_1|| + ||F_{r1}(G_i-G_i')||) + \frac{8B}{n^2}\\
\leq &\frac{8B}{n} + \frac{8B}{n} + \frac{8B}{n} + \frac{8B}{n^2}\\
\leq &\frac{32B}{n}
\end{align*}
The equality follows from the same argument as in proof of Lemma \ref{lem:gret}. The second inequality follows from triangle inequality, the third and fourth inequalities follows from Lemma \ref{lem:4pre}. Then by McDiarmid's inequality we complete the proof.
\end{proof}

\begin{proof}[Proof of Theorem \ref{thm:2}]
By (\ref{eq:1}), we have $||\hat{{\mathcal{R}}}_{YX_r}^{(n)} - (\Sigma_{YY} + \epsilon_nI)^{-1/2}\Sigma_{YX_r}(\Sigma_{X_rX_r} + \epsilon_nI)^{-1/2}||\leq I + II + III$, where
\begin{align*}
I &= ||\{(\hat{\Sigma}_{YY}^{(n)}+\epsilon_nI)^{-1/2} - ({\Sigma}_{YY} + \epsilon_nI)^{-1/2}\} \hat{\Sigma}_{YX_r}^{(n)}(\hat{\Sigma}_{X_rX_r}^{(n)} + \epsilon_nI)^{-1/2}||,\\
II &=||({\Sigma}_{YY}+\epsilon_nI)^{-1/2}\{\hat{\Sigma}_{YX_r}^{(n)}-{\Sigma}_{YX_r}\}(\hat{\Sigma}_{X_rX_r}^{(n)}+\epsilon_nI)^{-1/2}||,\\
III&=||({\Sigma}_{YY}+\epsilon_nI)^{-1/2}{\Sigma}_{YX_r}\{ (\hat{\Sigma}_{X_rX_r}^{(n)}+\epsilon_nI)^{-1/2} - ({\Sigma}_{X_rX_r} + \epsilon_nI)^{-1/2} ||.
\end{align*}

By (\ref{eq:DC}) we have
\begin{align*}
I = ||\{({\Sigma}_{YY} + \epsilon_nI)^{-1/2}(({\Sigma}_{YY} + \epsilon_nI)^{3/2} - (\hat{\Sigma}_{YY}^{(n)}+\epsilon_nI)^{3/2} ) + (\hat{\Sigma}_{YY}^{(n)} -  {\Sigma}_{YY}) \}\\
\times(\hat{\Sigma}_{YY}^{(n)}+\epsilon_nI)^{-3/2} \hat{\Sigma}_{YX_r}^{(n)}(\hat{\Sigma}_{X_rX_r}^{(n)} + \epsilon_nI)^{-1/2}||.
\end{align*}
From $|| ({\Sigma}_{YY} + \epsilon_nI)^{-1/2} ||\leq\epsilon_n^{-1/2}, ||(\hat{\Sigma}_{YY}^{(n)}+\epsilon_nI)^{-1/2} \hat{\Sigma}_{YX_r}^{(n)}(\hat{\Sigma}_{X_rX_r}^{(n)} + \epsilon_nI)^{-1/2}||\leq 1$, and Lemma \ref{lem:1}, 
\begin{align*}
I\leq \frac{1}{\epsilon_n}\left\{\frac{3}{\sqrt{\epsilon_n}}\max\{||{\Sigma}_{YY} + \epsilon_nI||,||\hat{\Sigma}_{YY}^{(n)} + \epsilon_nI||\} + 1\right\}||\hat{\Sigma}_{YY}^{(n)}-{\Sigma}_{YY} ||\leq cB^{1/2}\epsilon_n^{-3/2}||\hat{\Sigma}_{YY}^{(n)}-{\Sigma}_{YY} ||_{\HS}
\end{align*}
By Lemma \ref{lem:4} and Theorem \ref{thm:1}, we have
\be
\P\{||\hat{\Sigma}_{YY}^{(n)}-{\Sigma}_{YY} ||_{\HS}\geq c_1Bn^{-1/2} + t\}\leq \exp(-\frac{nt^2}{512B^2}).
\ee
Then 
\be
\P\{I\geq C_1'B^{3/2}n^{-1/2}\epsilon_n^{-3/2} + cB^{1/2}\epsilon_n^{-3/2}t\}\leq \exp(-\frac{nt^2}{512B^2}).
\ee
Using a similar argument, we have 
\be
\P\{III\geq C_1'B^{3/2}n^{-1/2}\epsilon_n^{-3/2} + cB^{1/2}\epsilon_n^{-3/2}t\}\leq \exp(-\frac{nt^2}{512B^2}).
\ee

Since $II\leq \epsilon_n^{-1}||\hat{\Sigma}_{YX_r}^{(n)}-{\Sigma}_{YX}||$,  we have $\P\{II \geq c_1Bn^{-1/2}\epsilon_n^{-1/2} + \epsilon_n^{-1/2}t\}\leq \exp(-\frac{nt^2}{512B^2})$. By the condition $\epsilon_n = o(1)$ we know that 
\be
\P\{II\geq C_1'B^{3/2}n^{-1/2}\epsilon_n^{-3/2} + c_3B^{1/2}\epsilon_n^{-3/2}t\}\leq \exp(-\frac{nt^2}{512B^2})
\ee
Let $t^* = C_1'B^{3/2}n^{-1/2}\epsilon_n^{-3/2} + c_3B^{1/2}\epsilon_n^{-3/2}t$. Then
\begin{align*}
&\P\{ ||\hat{{\mathcal{R}}}_{YX_r}^{(n)} - (\Sigma_{YY} + \epsilon_nI)^{-1/2}\Sigma_{YX_r}(\Sigma_{X_rX_r} + \epsilon_nI)^{-1/2}|| \geq 3t^* \}\\
\leq&\P\{I + II + III\geq 3t^*\}\\
\leq&\P\{I\geq t^*\} + \P\{II\geq t^*\} +\P\{III\geq t^*\}\\
\leq&3\exp(-\frac{nt^2}{512B^2})
\end{align*}
where the second inequality follows from the union bound.
Replace $3c_3B^{1/2}\epsilon_n^{-3/2}t$ by $u$, we have
\begin{align*}
\P\{ ||\hat{{\mathcal{R}}}_{YX_r}^{(n)} &- (\Sigma_{YY} + \epsilon_nI)^{-1/2}\Sigma_{YX_r}(\Sigma_{X_rX_r} + \epsilon_nI)^{-1/2}|| \\
&- C_1B^{3/2}n^{-1/2}\epsilon_n^{-3/2}\geq u\}\\
&\leq 3\exp(-\frac{\epsilon_n^3 nu^2}{512B^2})
\end{align*}
where $C_1 = 3 C_1'$.
\end{proof}

\begin{proof}[Proof of Theorem \ref{thm:3}]
First notice that $\{|\hat{\rho}_r(\epsilon_n) - \rho_r(\epsilon_n)| \geq cB^{3/2}\epsilon_n^{-3/2}n^{-\kappa}\}\subseteq \{||\hat{{\mathcal{R}}}_{YX_r}^{(n)} - (\Sigma_{YY} + \epsilon_nI)^{-1/2}\Sigma_{YX_r}(\Sigma_{X_rX_r} + \epsilon_nI)^{-1/2}|| \geq cB^{3/2}\epsilon_n^{-3/2}n^{-\kappa}\}$. Then by Theorem \ref{thm:2} we know for $0<\kappa<1/2$, there exist $C_3,C_4>0$, such that
\begin{align*}
\P\{|\hat{\rho}_r(\epsilon_n) - \rho_r{\epsilon_n}|\geq C_3B^{3/2}\epsilon_n^{-3/2}n^{-\kappa}\}\leq 3\exp(-{C_4Bn^{1-2\kappa}})
\end{align*}
Then by union bound we proved the first part of Theorem \ref{thm:3}. For the second part, we notice that if $\mathcal{M}\nsubseteq \hat{\mathcal{M}}$, then there must exist some $r\in \mathcal{M}$ such that $\hat{\rho}_r<C_3B^{3/2}\epsilon_n^{-3/2}n^{-\kappa}$. By condition (C2) we know that $\mathcal{M}\nsubseteq \hat{\mathcal{M}}$ implies $|\hat{\rho}_r(\epsilon_n) - \rho_r(\epsilon_n)|>C_3B^{3/2}\epsilon_n^{-3/2}n^{-\kappa}$ for some $r\in \mathcal{M}$. So we have
\begin{align*}
\P\{\mathcal{M}\subseteq \hat{\mathcal{M}}\} &= 1 - \P\{\mathcal{M}\nsubseteq \hat{\mathcal{M}}\}\\
&\geq 1 - \P\{\max_{r\in\mathcal{M}}|\hat{\rho}_r(\epsilon_n) - \rho_r(\epsilon_n)|>C_3B^{3/2}\epsilon_n^{-3/2}n^{-\kappa}\}\\
&\geq 1 - 3s\exp(-{C_4Bn^{1-2\kappa}})
\end{align*}
\end{proof}

\begin{lem}\label{lm: properties of additive kernel}
Suppose $m_{X_i}$ is the mean element of $\h_{X_i}$ for $i=1,\ldots,p$, and $\Sigma_{XY}$ is the covariance operator from $\h_Y$ to $\h_X$. Then we have
\begin{enumerate}
\item[(a)] $(m_{X_1},\ldots,m_{X_p})$ is the mean element of $\h_X$, denoted by $m_X$.
\item[(b)] $\| \Sigma_{XY} \|_{\HS}^2 = \sum_{i=1}^p\| \Sigma_{X_iY} \|_{\HS}^2$.
\end{enumerate}
\end{lem}
\begin{proof}
Assertion (a) follows from, for any $f=(f_1,\ldots,f_p)=\in\h_X$,
\begin{eqnarray*}
\langle f, m_X \rangle_{\h_X} = \sum_{i=1}^p \langle f_i, m_{X_i}\rangle_{\h_{X_i}} = \sum_{i=1}^p \E[f_i(X_i)] = \E\langle f, \kappa_X(\cdot,X) \rangle_{\h_X}.
\end{eqnarray*}
To show (b), by Lemma \ref{lem:gret},
\begin{eqnarray*}
\begin{split}
\|\Sigma_{YX}\|_{\HS}^2&=\E[\langle k_x(\cdot,X) - m_X, k_x(\cdot, \tilde{X}) - m_X\rangle_{\h_X}\langle k_y(\cdot,Y) - m_Y, k_y(\cdot, \tilde{Y}) - m_Y\rangle_{\mathcal{H}_Y}]\\
&= \E[\sum_{i=1}^p \langle k_{x_i}(\cdot,X_i) - m_{X_i}, k_{x_i}(\cdot, \tilde{X_i}) - m_{X_i}\rangle_{\h_{X_i}}\langle k_y(\cdot,Y) - m_Y, k_y(\cdot, \tilde{Y}) - m_Y\rangle_{\mathcal{H}_Y}]\\
&=\|\Sigma_{YX_i}\|_{\HS}^2.
\end{split}
\end{eqnarray*}
\end{proof}

Note \citet{fukumizu2007statistical} showed that $\Sigma_{XY}$ is Hilbert--Schmidt for any fixed $p$. Next, we extend their result to the case where $p$ grows to infinity.

\begin{lem}\label{lm: bounded covariance operator}
Suppose $\Sigma_{XX}$ is the covariance operator from $\h_X$ to $\h_X$. Then we have
\begin{eqnarray}\label{eq: lm: bounded covariance operator}
\|\Sigma_{YX}\|_{\HS}^2 = O[\lambda_{\mathrm{max}}(\Sigma_{XX})].
\end{eqnarray}
\end{lem}

\def\tr{\mathrm{tr}}
\begin{proof}
First note that $\|\Sigma_{YX}\|_{\HS}^2$ is bounded by
\begin{eqnarray*}
\|\Sigma_{YX}\|_{\HS}^2 \leq \|\Sigma ^{1/2}_{XX} \|^2\cdot \|\Sigma^{1/2}_{YY} \|^2_{\HS}.
\end{eqnarray*}
Then it suffices to show that $\|\Sigma^{1/2}_{YY} \|^2_{\HS} = \tr(\Sigma_{YY}){<\infty}$. By definition $\tr(\Sigma_{YY})$ is equal to
\begin{eqnarray*}
\begin{split}
\tr(\Sigma_{YY}) = \E[ \| \kappa_Y(\cdot,Y)-m_Y \|^2]=\E\kappa_Y(Y,Y) - \|m_Y \|^2,
\end{split}
\end{eqnarray*}
which is finite by Lemma \ref{lem:4pre}. The proof is completed.
\end{proof}

\begin{proof}[Proof of Theorem \ref{thm:4}]
%By the orthogonal decomposition that 
%$$
%\Sigma_{YY} = \Sigma_{YX}\Sigma_{XX}^{-1}\Sigma_{XY}  + (\Sigma_{YY} -\Sigma_{YX}\Sigma_{XX}^{-1}\Sigma_{XY} )
%$$
%we know that $||\Sigma_{YX}\Sigma_{XX}^{-1}\Sigma_{XY} ||\leq ||\Sigma_{YY}|| = O(1)$. On the other hand, define $M = \Sigma_{XX}^{-1/2}\Sigma_{XY}$, we have
%$$
%||\Sigma_{XY}||^2 \leq ||\Sigma_{XX}^{1/2}||^2||M||^2 = ||\Sigma_{XX}^{1/2}||^2 || \Sigma_{YX}\Sigma_{XX}^{-1/2}\Sigma_{XY} ||\leq O(\lambda_{\max}(\Sigma_{XX})).
%$$ 
By definition, $\rho_r(\epsilon_n) = || (\Sigma_{YY} + \epsilon_nI)^{-1/2}\Sigma_{YX_r}(\Sigma_{X_rX_r}+\epsilon_nI)^{-1/2} ||\leq \epsilon_n^{-1}||\Sigma_{X_rY}||$. Define $\Sigma_{XY} = (\Sigma_{X_1Y},...,\Sigma_{X_pY})$, then
\be\label{eq:boundrho}
\sum_{r=1}^p\rho_r^2(\epsilon_n) \leq \sum_{r=1}^p\epsilon_n^{-2}||\Sigma_{X_rY}||_{\HS}^2 \leq \epsilon_n^{-2}||\Sigma_{XY}||_{\HS}^2 = O(\epsilon_n^{-2}\lambda_{\max}(\Sigma_{XX}))
\ee
The second last inequality follows from Lemma \ref{lm: properties of additive kernel}, and the last equality follows from Lemma \ref{lm: bounded covariance operator}.
This implies that the number of $\{r:\rho_r(\epsilon_n) > \eta\epsilon_n^{-1}n^{-\kappa}\}$ cannot exceed $O(n^{2\kappa}\lambda_{\max}(\Sigma_{XX}))$ for any $\eta > 0$, which implies $|\{r:\rho_r(\epsilon_n) > \eta\epsilon_n^{-3/2}n^{-\kappa}\}|\leq O(n^{2\kappa}\lambda_{\max}(\Sigma_{XX}))$ for any $\eta> 0$ because $\epsilon_n \leq 1$. Thus, on the set
$$
B_n = \left\{ \max_{1\leq r\leq p}|\hat{\rho}_r(\epsilon_n) - \rho_r(\epsilon_n) |\leq \eta \epsilon_n^{-3/2}n^{-\kappa} \right\},
$$
the number of $\{r:\hat{\rho}_r(\epsilon_n)>2\eta\epsilon_n^{-3/2}n^{-\kappa}\}$ cannot exceed the number of $\{r:\rho_r(\epsilon_n) > \eta\epsilon_n^{-3/2}n^{-\kappa}\}$, which is bounded by $O(n^{2\kappa}\lambda_{\max}(\Sigma_{XX}))$. By taking $\eta = C_3/2$, we have
$$
\P\{|\hat{\mathcal{M}}\leq O(n^{2\kappa}\lambda_{\max}(\Sigma_{XX}))|\}\geq \P B_n
$$
The conclusion follows from Theorem \ref{thm:3}.
\end{proof}

\vskip 0.2in
\section*{{*}References}
\nobibliography{April25}
\begin{enumerate}
\item\bibentry{akaho2006kernel}
\item\bibentry{aronszajn1950theory}
\item\bibentry{bach2003kernel}
\item\bibentry{baker1973joint}
\item\bibentry{balasubramanian2013ultrahigh}
\item\bibentry{candes2007dantzig}
\item\bibentry{fan2001variable}
\item\bibentry{fan2008sure}
\item\bibentry{fan2009ultrahigh}
\item\bibentry{fan2010selective}
\item\bibentry{fan2011nonparametric}
\item\bibentry{fukumizu2007kernel}
\item\bibentry{fukumizu2007statistical}
\item\bibentry{fukumizu2009kernel}
\item\bibentry{gretton2005measuring}
\item\bibentry{ji2012ups}
\item\bibentry{kang2011spatio}
\item\bibentry{lee2013general}
\item\bibentry{lee2016variable}
\item\bibentry{li2014additive}
\item\bibentry{li2012feature}
\item\bibentry{mcdiarmid1989method}
\item\bibentry{melzer2001nonlinear}
\item\bibentry{micchelli2006universal}
\item\bibentry{reed1980methods}
\item\bibentry{sejdinovic2013equivalence}
\item\bibentry{song2012feature}
\item\bibentry{tibshirani1996regression}
\item\bibentry{willsey2013coexpression}
\item\bibentry{yamanishi2003extraction}
\item\bibentry{zhang2010nearly}
\item\bibentry{zhao2006model}
\item\bibentry{zeng2012large}
\end{enumerate}

\end{document}